\documentclass[correspondance]{IEEEtran}

\usepackage[dvipsnames, svgnames, x11names]{xcolor}
\usepackage{graphicx}
\usepackage{cite}
\usepackage{bm}
\usepackage{amsmath}
\usepackage{amssymb}
\usepackage{amsfonts}
\usepackage{amsthm}
\usepackage{algorithm} 
\usepackage{algorithmic}
\usepackage{array}
\usepackage{epstopdf}
\usepackage{multicol}
\usepackage{makecell}
\usepackage{mathrsfs,amsmath}
\usepackage{booktabs}
\usepackage{booktabs}
\usepackage{graphicx}
\usepackage{float}
\usepackage[caption=false,font=normalsize,labelfont=sf,textfont=sf]{subfig}
\usepackage{cite}
\begin{document}
	\title{\huge Processing Load Allocation of On-Board Multi-User 
		Detection for Payload-Constrained Satellite Networks}

\author{Sirui~Miao,
	Neng~Ye,~\IEEEmembership{Member,~IEEE,}
	Peisen~Wang,
	and~Qiaolin~Ouyang
	
	\vspace{-0.75em}
	
	\thanks{This research was funded by the National Natural Science Foundation of
		China under grant 62101048, and Young Elite Scientists Sponsorship Program
		by CAST under grant 2022QNRC001. (Corresponding author: Neng Ye.)}
\thanks{S. Miao and N. Ye are with the School of Cyberspace Science and Technology, Beijing Institute of Technology, Beijing 100081, China (E-mails:siruimiao@bit.edu.cn,ianye@bit.edu.cn).}
\thanks{P. Wang and Q. Ouyang are with the School of Information and Electronics, Beijing Institute of Technology, Beijing 100081, China (E-mails:wps@bit.edu.cn, qlouyang@bit.edu.cn).}
}
\markboth{Submitted to IEEE Wireless Communication Letters}{Yu \MakeLowercase{Yang{et al.}}: Processing Load Allocation of On-Board Multi-User Detection for Payload-Constrained Satellite Networks}

\maketitle
	\maketitle
	\begin{abstract}
		The rapid advance of mega-constellation facilitates the booming of direct-to-satellite massive access, where
		 multi-user detection  is critical to alleviate the induced inter-user interference.
		 While centralized implementation of on-board  detection induces unaffordable complexity for a single satellite, this paper proposes to allocate the processing load among cooperative satellites for finest exploitation of distributed processing power.
		 Observing the inherent disparities among users, we first excavate the closed-form trade-offs between achievable sum-rate and the processing load corresponding to the satellite-user matchings, which leads to a system sum-rate maximization problem under stringent payload constraints.
		 To address the non-trivial integer matching, we develop a quadratic transformation to the original problem, and prove it an equivalent conversion. The problem is further simplified into a series of subproblems employing successive lower bound approximation which obtains polynomial-time complexity and converges within a few iterations.
		 Numerical results show remarkably complexity reduction compared with centralized processing, as well as around 20\% sum-rate gain compared with other allocation methods.
		
	\end{abstract}
	\section{Introduction}
	Direct-to-satellite Internet of Things (IoT) is regarded as a charming paradigm to accommodate massive ubiquitous connections in future 6G \cite{Qingservice}. The numerous IoT terminals underneath each satellite beam, however, cause severe inter-user interference, either due to the uncoordinated access, or the time and frequency asynchrony due to large channel dynamics. Multi-user detection (MUD) is a critical technology to suppress interference and promote system throughput. Combined with multi-satellite cooperative reception, which is facilitated via mega-constellation, MUD can further enhance the system performance by excavating spatial diversity \cite{emir2021deepmud}.
	
	One straightforward approach embracing distributed satellites is to collect the received signals and perform the centralized  processing in a single computing fabric, which is similar to cloud-computing in terrestrial networks \cite{wang2020grant,9174777,9720914}. Ke \textit{et al.} in \cite{9174777} proposed a successive interference cancellation based active user detection algorithm where the distributed units sample and transfer the received signals to a central unit. To reduce the detection complexity of centralized  processing, Wang \textit{et al.} in \cite{9720914} proposed a symbol-wise  maximum likelihood detection algorithm using  sparsity of the channel matrix. Unfortunately, the centralized implementation is impractical in the satellite platform due to the unbearable computational complexity at the centralized processing unit \cite{malkowsky2017world}.

	An alternative approach is the decentralized processing that 
	distributes the processing load across multiple receiving points, with a central receiver combining the local estimations \cite{zhang2020decentralized,amiri2021distributed,hayakawa2018distributed}.  To lower the computational cost, Amiri \textit{et al.} in \cite{amiri2021distributed} proposed a set of receiver options for complexity adjustable distributed detection based on variational message passing algorithm. 
	Hayakawa \textit{et al.} in \cite{hayakawa2018distributed}  further reduced the complexity in the view of system architecture by proposing a modified approximate message passing algorithm  where fusion center is no longer required.
	Nonetheless, previous research does not consider the individual processing load constraint for a single unit, which forms the stringent limitation of the payload-constrained satellite networks. Therefore, it calls for  sophisticated design on complexity control for a satellite to implement distributed signal processing \cite{liu2015analytical}.
	
	Usually, the complexity of MUD increases exponentially with the user number, which inspires us to excavate the number of serving users for the satellite.
	Compared with terrestrial communications, satellite networks experience prominent disparities among users with respect to
	 channel gains.  
	 Huge computation power is consumed to recover weak users' signals but doesn't yield significant gain in system sum-rate. These  characteristics motivate us to elaborately  design satellite-user matching relationship to meet the payload-constrained satellite networks while guaranteeing the system sum-rate.

	In this paper, we introduces a novel processing load allocation method for spaceborne MUD in payload-constrained distributed satellite networks. First, the closed-form trade-offs of both the system sum-rate and the processing load with respect to the satellite-user matchings are developed to convert load allocation into sum-rate optimization. Then, we formulate the processing load allocation problem as a  sum-rate maximization problem with respect to  satellite-user matchings. To address the non-trivial integer matching, an equivalent transformation is developed using
	quadratic transformation as well as Lagrangian dual function, which is further converted into a series of subproblems using successive lower bound approximation. We prove that the  obtained results are the Karush-Kuhn-Tucker (KKT) points of the original problem and obtain polynomial-time complexity.
	Numerical results show 75\% complexity reduction with less than 10\% sum-rate loss compared with centralized processing. Compared to other allocation methods, our proposed allocation method provides around 20\% sum-rate gain.

	\section{System Model and Problem Formalization}
	\newtheorem{proposition}{Proposition}
	In this section, we give the signal models of the transceivers and formulate the load allocation problem as a sum-rate maximization problem by optimizing satellite-user matchings.

	\subsection{Signal Model}
	\begin{figure}[t]
		\centering
		\includegraphics[width=1.0\linewidth]{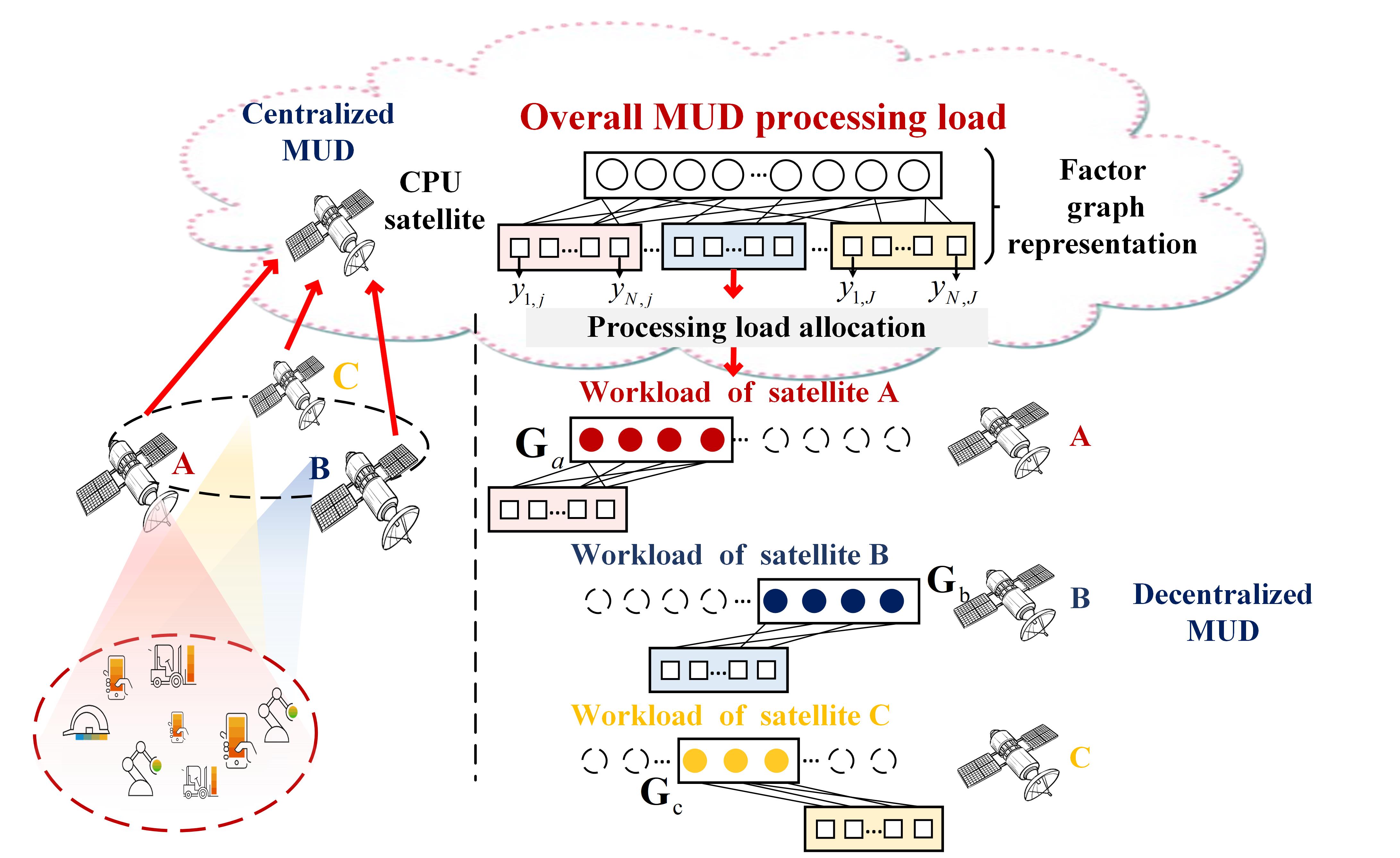}
		\caption{Illustration of processing load allocation in cooperative network.}
		\label{framework_short}
	\end{figure}
	We consider an uplink transmission scenario where $K$ ground terminals transmit their source symbols over $N$ resource elements (REs) in a non-orthogonal way to enhance the connectivity \cite{fang2022noma}. These terminals are simultaneously served by $J$ satellites interconnected via ISLs.  
	Denote $\mathcal{X}$ as the modulation constellation alphabet with $\left|\mathcal{X}\right|=M$ standing for the cardinality. For each symbol interval, every ${\log_2}M$ source bits of the $k$-th user are first modulated to $x_k \in \mathcal{X}$ according to the bit-to-symbol mapping, and then spread over $N$ REs according to user-specific spreading sequence $\mathbf{s}_k=[s_{1,k},s_{2,k},\ldots,s_{N,k}]^T \in \mathbb{C}^N$ with $\left\| {\mathbf{s}_k} \right\|=1$.
	Denote $F(n)=\left\{k \mid s_{n,k}\ne 0, k=1,\ldots,K \right\}$ as the set of users that occupy the $n$-th RE. Let
	$d_f$ denotes  the maximum number of users that are allowed to interfere within a single RE.
	
	Denote $h_{n,k,j}$ as the channel coefficient between the $j$-th satellite and the $k$-th user on the $n$-th RE
	\begin{equation}
		{h_{n,k,j}} = \frac{{{g_{n,k,j}}}}{{D\left( {{d_{k,j}}} \right)}},
	\end{equation}
	where $g_{n,k,j}$ and  $d_{k,j}$ respectively denote the complex fading gain and  the distance between the $j$-th satellite and the $k$-th user, and $D(\cdot)$ denotes the path loss function.
	
	The received signal at the $j$-th satellite over $N$ REs, denoted as $\mathbf{y}_j=[y_{1,j},y_{2,j},\ldots,y_{N,j}]^T\in \mathbb{C}^N$, is thus given by
	\begin{equation}
		\begin{aligned}
			{{\mathbf{y}}_j} = \sum\limits_{k = 1}^K {{\rm{diag}}({{\bf{h}}_{k,j}}){\mathbf{s}_k}{x_k} + {{\bf{w}}_j}},
		\end{aligned}
	\end{equation}
	where $\mathbf{w}_j \sim \mathcal{C N} \left( 0,\sigma^2 \mathbf{I}\right)$ is the additive complex Gaussian noise at the $j$-th satellite and $\mathbf{h}_{k,j}=[h_{1,k,j},\ldots,h_{N,k,j}]^T$.

	\subsection{Problem Statement}
	\subsubsection{Analysis of MUD Processing Load}
	Suppose that $J$ satellites aim to jointly estimate the transmit signals $\mathbf{x}=[x_1,x_2,\ldots,x_K]$ by following the maximum a \textit{posteriori} probability principle which can be written as
	\begin{equation}
		\label{equa_ML}
		\begin{aligned}
			\hat{x}_{k}=&\underset{a \in \mathcal{X}}{\arg \max } \sum\limits_{\scriptstyle{\bf{x}} \in {{\cal X}^K}\atop
				\scriptstyle{x_k} = a}  {\prod\limits_{k = 1}^K {P({x_k})}}
			 \prod\limits_{j = 1}^J {\prod\limits_{n = 1}^N p{\left( {{y_{n,j}}\mid \left\{x_l\right\}_{l \in F(n)}} \right)} } ,
		\end{aligned}
	\end{equation}
	where $\left\{x_l\right\}_{l \in F(n)}$ represents the set of transmit symbols of users occupying the $n$-th RE. As illustrated in left-hand-side of Fig. \ref{framework_short}, with the centralized processing, all received signal $\mathbf{y}_j$ should be aggregated to a CPU satellite to perform centralized detection which poses huge computational burden.
	
	Due to the stringent limitation on the processing capability of satellites, the MUD processing load  should be distributed onto multiple satellites for cooperative detection. To reduce the product operation and carry out the MUD distributedly, \eqref{equa_ML} can be written in the form of log-likelihood ratio (LLR) as
	\begin{equation}
		\label{equa_ML_LLR_decentralized}
		{L_j}({x_k}) = \sum\limits_{n = 1}^N {\log \frac{{p\left( {{y_{n,j}}\mid {x_k},{{\left\{ {{x_l}} \right\}}_{l \in F(n)\backslash k}}} \right)}}{{p\left( {{y_{n,j}}\mid \tilde x,{{\left\{ {{x_l}} \right\}}_{l \in F(n)\backslash k}}} \right)}}}  + \sum\limits_{\substack{j' = 1,\ldots,J \\ j' \ne j}} {{L_{j'}}({x_k})},
	\end{equation}
	where ${\tilde x}$ is a fixed reference symbol in $\mathcal{X}$ and ${L_j}({x_k})$ is the estimated LLR of  $x_k$ made by the $j$-th satellite. The entire MUD are decomposed in a round robin method where each satellite makes local estimation and output temporary ${L_j}({x_k})$ based on information transmitted from other satellites. The first term of \eqref{equa_ML_LLR_decentralized} is the local estimation  made by the $j$-th satellite and the second term represents the temporary joint estimation result made by other satellites.	Based on the superimposed structure, the first term of \eqref{equa_ML_LLR_decentralized} can be expressed in a sparse factor graph as shown in upper portion of Fig. \ref{framework_short}.

	\begin{proposition}\label{pro}
		The  MUD processing load at the $j$-th satellite increases exponentially with the number of its target users.
	\end{proposition}
	\begin{proof}
		As  observed in \eqref{equa_ML_LLR_decentralized}, to calculate the  $	{L_j}({x_k})$, the $j$-th satellite needs to traverse all possible values of $\left\{ {{x_l}} \right\} \in \mathcal{X} $ for all ${l \in F(n)\backslash k}$. Therefore, a total of $M^{d_f-1}$ possible combinations need to be calculated which is exponentially increased with $d_f$. Notice that $d_f$ is
		the number of target detected user of the $n$-th RE of the $j$-th satellite, thus the MUD processing load of the $j$-th satellite increases exponentially with the number of its target users to be detected.
	\end{proof}
\newtheorem{definition}{Definition} 
	\begin{definition}
		Define an indicating matrix $\mathbf{Q} \in \left\{0,1\right\}^{K\times J}$ to represent the satellite-user matchings where the element on the $k$-th row and $j$-th column denotes whether the source symbol of the $k$-th user is recovered at the $j$-th satellite. Also, define $\mathbf{G}_j=\left\{ k \mid q_{k,j} =1\right\}$ and $\mathbf{\bar{G}}_j=\left\{ k \mid q_{k,j} =0\right\}$ as the target and non-target user set of $j$-th satellite, respectively.
	\end{definition}
	\subsubsection{Impact of Satellite-User Matchings on System Sum-Rate}
	Since the source symbols of users in  $\mathbf{\bar{G}}_j$ are not recovered at the $j$-th satellite, their signals must be estimated from other satellites and canceled from the received signals $\mathbf{y}_j$. However, perfect cancellation is not achievable due to information loss or inaccurate estimation by other satellites. Therefore, for the  $k$-th user on the $n$-th RE of the $j$-th satellite, it suffers from additional interference caused by imperfect cancellation as
	\begin{equation}
		\label{interference_original}
		I_{n,k,j}=\epsilon \sum\limits_{i = 1,|h_{n,i,j}|^2<|h_{n,k,j}|^2}^K {{{\left| {{h_{n,k,j}}} \right|}^2} \cdot |{s_{n,i}}|^2 \cdot {q_{i,j}}},
	\end{equation}
	where $\epsilon=E\left\{ |x_i-{\tilde x_i}|^2\right\}$ is the expectation of the difference between actual signal $x_i$ and the estimated signal $\tilde{x}_i$. 
	
	As observed from \eqref{interference_original}, the interference received by the $k$-th user on the $j$-th satellite is associated with $\mathbf{G}_j$, which means $\mathbf{G}_j$ is directly linked to the sum-rate of the $k$-th user.
	Assume  $|h_{n,1,j}|^2\ge|h_{n,k,j}|^2\ge\ldots|h_{n,K,j}|^2$, the overall sum-rate $C_{\text{SUM}}$ of the system   can be obtained as
	\begin{equation}
		\label{function1}
		\begin{aligned}
			&C_{\text{SUM}}=\\&\sum\limits_{n = 1}^N {\sum\limits_{j = 1}^J {\sum\limits_{k = 1}^K {{\log _2 \Bigg(1+\frac{{{{\left| {{h_{n,k,j}}} \cdot {s_{n,k}} \right|}^2}  \cdot {q_{k,j}}}}{{\sigma _n^2 + \sum\limits_{i = 1}^{k - 1} {{{\left| {{h_{n,i,j}}} {s_{n,i}} \right|}^2} \cdot {q_{i,j}} + I_{n,k,j} } }}\Bigg)}} } }. 
		\end{aligned}
	\end{equation}
	
	The coupling relationship between the individual processing load and the system sum-rate with respect to $\mathbf{Q}$ inspire us to optimize the processing load of MUD by properly assigning the user matching of the satellites.
	
	Therefore, we formulate the MUD processing load allocation problem by maximizing the system sum-rate for payload-constrained satellites, as follows
	\begin{equation}
		\label{function2}
		\begin{aligned}
			\mathcal{P}_{\text{MAX-SR}}: & \hfill \,
			{\max _{\mathbf{Q}}} \, C_{\text{SUM}}\\
			\textbf{ s.t. }&
			\mathbf{C}_1:\,\sum\limits_{j = 1}^J {{q_{k,j}}}  \ge {q_l}, \,1 \le k \le K,\\
			&\mathbf{C}_2: \,\sum\limits_{k = 1}^K {{q_{k,j}}}  \le {q_s}, \,1 \le J \le J,\\
			& \mathbf{C}_3:\,q_{k,j}\in\left\{0,1\right\}, \,1 \le k \le K, \,1 \le J \le J,
		\end{aligned}
	\end{equation}
	where $\mathbf{C}_1$ constrains the least number of satellites allocated per user to guarantee the detection accuracy, $\mathbf{C}_2$ constrains the maximum number of users allocated per satellite to control the processing load at each satellite, and $\mathbf{C}_3$ denotes the binary constraint on  $q_{k,j}$.

	It should be noticed that $\mathcal{P}_{\text{MAX-SR}}$ is not a standard convex problem as the result of integer constraint in $\mathbf{C}_3$ and the interference term on the denominator of the objective function.

	\section{Proposed Workload Allocation Method for Distributed Multi-User Detection}
	To solve $\mathcal{P}_{\text{MAX-SR}}$, we aim to transform it into a convex problem with respect to $\mathbf{Q}$. We first introduce Lagrangian dual function and quadratic transformation to conduct equivalent conversion. Then, successive lower bound approximation (SLM) is used to decouple the converted problem into approximated subproblems to perform iterative convex optimization.

	\subsection{Equivalent Transformation of $\mathcal{P}_{\text{MAX-SR}}$}
	We first relax the integer constraint on  $\left\{q_{k,j}\right\}$ and allow them to vary continuously $0\le\tilde{q}_{k,j}\le 1$. A penalty function is added to $C_{\text{SUM}}$ such that non-integer solutions of $\mathbf{Q}$ are penalized. We choose a moderately high value $\lambda$ and let $p(\mathbf{Q})=\lambda \sum\limits_{j = 1}^J {\sum\limits_{k = 1}^K {(\tilde{q}_{k,j}^2 - {\tilde{q}_{k,j}})} } $ be the penalty function since $p(\mathbf{Q})<0$ for all non-integer solutions of $\mathbf{Q}$. $\mathcal{P}_{\text{MAX-SR}}$ is transformed into 
	$\mathbf{\mathcal{P}}_{1}(\mathbf{Q})$:
	\begin{equation}
		\label{function4}
		\begin{aligned}
			\mathbf{\mathcal{P}}_{1}(\mathbf{Q}): \, & 
			{\max _{\mathbf{Q}}} \, C_{\text{SUM}}+p\left( {{\bf{Q}}} \right)\\
			\textbf{ s.t. }&
			\mathbf{C}_1,\mathbf{C}_2,\\
			&\mathbf{C}_4:\, 0\le \tilde{q}_{k,j}\le 1,\,\,1 \le k \le K, \,1 \le J \le J. 
		\end{aligned}
	\end{equation}
	
	The objective function of $\mathbf{\mathcal{P}}_{1}(\mathbf{Q})$ is still non-concave due to the interference term, hence we introduce auxiliary matrix $\boldsymbol{\gamma}$ with element $\left\{\gamma_{n,k,j}\right\}$. Then, we transform $\mathbf{\mathcal{P}}_{1}(\mathbf{Q})$ into an equivalent problem $\mathbf{\mathcal{P}}_2(\mathbf{Q},\boldsymbol{\gamma})$ as
	\begin{equation}
		\label{function5}
		\begin{aligned}
			\mathcal{P}_2(\mathbf{Q},\boldsymbol{\gamma}):& \mathop {\max }\limits_{{\bf{Q}},{\boldsymbol{\gamma }}} \sum\limits_{n = 1}^N {\sum\limits_{j = 1}^J {\sum\limits_{k = 1}^K {\log_2 (1 + {\gamma _{n,k,j}})} } }  + p\left( {{\bf{Q}}} \right)\\
			\textbf{ s.t. }&
			\mathbf{C}_1,\mathbf{C}_2,\mathbf{C}_4,\\
			& \mathbf{C}_5 : \gamma_{n,k,j}\le\frac{{{{\left| {{h_{n,k,j}}} {s_{n,k}} \right|}^2} \cdot {\tilde{q}_{k,j}}}}{{\sigma _n^2 + \sum\limits_{i = 1}^{k - 1} {{{\left| {{h_{n,i,j}}} {s_{n,i}} \right|}^2}  \cdot {\tilde{q}_{i,j}} + I_{n,k,j} } }}
		\end{aligned}
	\end{equation}
	
	Given $\mathbf{Q}$, \eqref{function5} is a convex problem with respect to  $\boldsymbol{\gamma}$. The corresponding Lagrangian function with dual variables $\left\{\mu_{n,k,j} \right\}$ for $\mathbf{C}_5$ can be written as 
	\begin{equation}\label{function6}
		\begin{aligned}
			& L_P(\boldsymbol{\gamma}, \boldsymbol{\mu})=\sum_{n=1}^N \sum_{j=1}^J \sum_{k=1}^K \log _2\left(1+\gamma_{n, k, j}\right) \\
			& -\sum_{n=1}^N \sum_{j=1}^J \sum_{k=1}^K \mu_{n, k, j}\left(\gamma_{n, k, j}-\right. \\
			& \left.\frac{\left|h_{n, k, j}\cdot s_{n, k}\right|^2  \cdot \tilde{q}_{k, j}}{\sigma_n^2+\sum_{i=1}^{k-1}\left|h_{n, i, j} \cdot {s_{n,i}}\right|^2  \cdot \tilde{q}_{i, j}+I_{n, k, j}}\right)+p(\mathbf{Q}).
		\end{aligned}
	\end{equation}
	Applying  the KKT conditions, we have the optimal dual variables as
	\begin{equation}
			\label{function10}
				\mu_{n,k,j}^*=\frac{1}{\ln 2 *\left(1+\gamma_{n, j, k}^*\right)}
	\end{equation} 
where
\begin{equation}
	\gamma_{n, k, j}^*=\frac{\left|h_{n, k, j}\right|^2 \cdot\left|s_{n, k}\right|^2 \cdot q_{k, j}}{\sigma_n^2+\sum_{i=1}^{k-1}\left|h_{n, i, j}\cdot s_{n, k}\right|^2  \cdot q_{i, j}+I_{n, k, j}}
\end{equation}
	Then we substitute \eqref{function10} into \eqref{function6} and $	\mathcal{P}_2(\mathbf{Q},\boldsymbol{\gamma})$ can be transformed into 
	\begin{equation}\label{function7}
		\begin{aligned}
			\mathbf{\mathcal{P}}^{'}_2(\mathbf{Q},\boldsymbol{\gamma}):&
			\max _{\boldsymbol{\gamma}}  \sum_{j=1}^J \sum_{n=1}^N \sum_{k=1}^K\left[\left(\log _2\left(1+\gamma_{n, k,j}\right)-\frac{\gamma_{n, k,j}}{\ln 2}\right)\right. \\
			+ & \left.\frac{{{{\left| {{h_{n,k,j}}}\cdot {s_{n,k}}\right|}^2}\cdot{\tilde{q}_{k,j}}(1 + {\gamma _{n,k,j}})}}{{\ln 2(\sigma _z^2 + \sum\limits_{i = 1}^k {{{\left| {{h_{n,i,j}}} \cdot {s_{n,i}} \right|}^2}} {\tilde{q}_{i,j}} + I_{n,k,j})}}\right] + p(\mathbf{Q})\\
			\textbf{ s.t. }&
			\mathbf{C}_1,\mathbf{C}_2,\mathbf{C}_4.
		\end{aligned}
	\end{equation}
	\newtheorem{lemma}{Lemma}
	\begin{lemma}\label{lemma1}
		$\mathcal{P}^{'}_2(\mathbf{Q},\boldsymbol{\gamma})$ is equivalent to $\mathcal{P}_2(\mathbf{Q},\boldsymbol{\gamma})$.
	\end{lemma}
	\begin{proof}
		We resort to KKT conditions to demonstrate the equivalence. \eqref{function10} is obtained using the Lagrangian stationarity and complementary slackness. Since  $\tilde{q}_{k,j}\ge 0$ always holds in $\mathbf{C}_4$, the dual feasibility holds as well.
	\end{proof}

	\subsection{Iterative Solution Using Successive Lower-bound Approximation}
	
	For fixed $\boldsymbol{\gamma}$, the objective function of $\mathcal{P}^{'}_2(\mathbf{Q},\boldsymbol{\gamma})$ is not a standard concave function in terms of  $\left\{q_{k,j}\right\}$ but is a sum-of-ratios problem. Therefore, we can rewrite the optimization problem $\mathcal{P}_3(\mathbf{Q},\boldsymbol{\gamma})$ by quadratic transformation by introducing auxiliary variables $\left\{\theta_{n,k,j}\right\}$, as
	\begin{equation}\label{function8}
		\begin{aligned}
			\mathbf{\mathcal{P}}_3(\mathbf{Q},\boldsymbol{\gamma,\theta})&:\\
			\max _{\boldsymbol{\theta}} & \sum_{j=1}^J\sum_{n=1}^N \sum_{k=1}^K\left[2 \theta_{n, k, j} \sqrt{\left| {{h_{n,k,j}}}  {s_{n,k}}\right|^2 \tilde{q}_{k, j}\left(1+\gamma_{n, k,j}\right)}\right. \\
			& \left.-\theta_{n, k,j}^2\left({\sigma _n^2 + \sum\limits_{i = 1}^k {{{\left| {{h_{n,i,j}}}{s_{n,i}} \right|}^2}}{\tilde{q}_{i,j}} + I_{n,k,j}}\right) \ln 2\right.\\&\left.+g(\boldsymbol{\gamma}, \boldsymbol{\theta})\right]+p(\mathbf{Q})\\
			\textbf{ s.t. }&
			\mathbf{C}_1,\mathbf{C}_2,\mathbf{C}_4.
		\end{aligned}
	\end{equation}
	where $g(\boldsymbol{\gamma})=\log_2\left( {1 + {\gamma _{n,k,j}}} \right) - \frac{{{\gamma _{n,k,j}}}}{{\ln 2}}$.
	
	\renewcommand{\algorithmicrequire}{ \textbf{Input:}} 
	\renewcommand{\algorithmicensure}{ \textbf{Output:}} 
	\begin{algorithm}[t]
		\caption{Proposed Load Allocation Method for $\mathcal{P}_{\text{MAX-SR}}$ }\label{algorithm1}
		\begin{algorithmic}[1]
			\REQUIRE  
			$\mathbf{H}$, $q_s$, convergence tolerance $\varepsilon_m$,$\varepsilon_n$.
			\ENSURE  Allocation matrix $\mathbf{Q}$.
			\STATE \textbf{Initialize:} Set the initial value of allocation matrix $\mathbf{Q}^{(0,0)}$ that meets $\mathbf{C}_1,\mathbf{C}_2,\mathbf{C}_4$. Set $m=0$ and $n=0$.
			\REPEAT
			\STATE Compute optimized variables $\boldsymbol{\gamma}^{(m)}$ by \eqref{function10};
			\STATE Given $\boldsymbol{\gamma}$ and $\mathbf{Q}$, update $\boldsymbol{\theta}^{(m,n)}$ solving $\mathbf{\mathcal{P}}_3(\mathbf{Q},\boldsymbol{\gamma,\theta})$;
			\REPEAT
			\STATE Given $\boldsymbol{\gamma}$ and $\boldsymbol{\theta}$, update $\mathbf{Q}^{(m,n)}$ solving $\mathbf{\mathcal{P}}^{'}_3(\mathbf{Q},\boldsymbol{\gamma,\theta})$;
			\STATE $n=n+1$;
			\UNTIL{reach convergence $\varepsilon_n$}
			\STATE $m=m+1$;
			\UNTIL{reach convergence $\varepsilon_m$}.
		\end{algorithmic}
	\end{algorithm}

	However, for fixed $\left\{\theta_{n,k,j}\right\}$,  $\mathbf{\mathcal{P}}_3(\mathbf{Q},\boldsymbol{\gamma,\theta})$ is still non-concave in terms of  $\left\{q_{k,j}\right\}$ due to the existence of the penalty term $p(\mathbf{Q})$. To solve $\mathbf{\mathcal{P}}_3(\mathbf{Q},\boldsymbol{\gamma,\theta})$, we resort to SLM algorithm following the instructions given in \cite{razaviyayn2013unified}.
	\begin{equation}
		\label{function9}
		\begin{aligned}
			\mathbf{\mathcal{P}}^{'}_3(\mathbf{Q},\boldsymbol{\gamma,\theta})&:\\
			\max _{\mathbf{Q}} & \sum_{j=1}^J\sum_{n=1}^N \sum_{k=1}^K\left[2 \theta_{n, k, j} \sqrt{\left|h_{n, k}\right|^2 \tilde{q}_{k, j}\left(1+\gamma_{n, k,j}\right)}\right. \\
			-\theta_{n, k,j}^2 & \left.\left({\sigma _z^2 + \sum\limits_{i = 1}^k {{{\left| {{h_{n,i,j}}} \right|}^2}} |{s_{n,i}}|^2{\tilde{q}_{i,j}} + I_{n,k,j}}\right) \ln 2 \right.\\
			+g(\boldsymbol{\gamma},&\left. \boldsymbol{\theta})\right]+p \left( {{{\bf{Q'}}}} \right)+ {\rm{tr}}\left[ {\nabla p {{\left( {{{\bf{Q}}'}} \right)}^T}\left( {{{\bf{Q}}  } - {{\bf{Q}}'}} \right)} \right] \\
			\textbf{ s.t. }&
			\mathbf{C}_1,\mathbf{C}_2,\mathbf{C}_4.
		\end{aligned}
	\end{equation}
	where ${\left[ {\nabla p \left( {{\mathbf{Q}}} \right)} \right]_{k,j}}$ is a $K \times J$ matrix that satisfies
	\begin{equation}
		{\left[ {\nabla p \left( {{\mathbf{Q}}} \right)} \right]_{k,j}} = {\left. {\frac{{\partial p (\mathbf{Q})}}{{\partial {\tilde{q}_{k,j}}}}} \right|_{\mathbf{Q} = {\mathbf{Q}'}}},
	\end{equation}
	\begin{lemma}\label{lemma2}
		The value of $\mathbf{Q}$ under fixed $\boldsymbol{\gamma}$ and $\boldsymbol{\theta}$ can be optimized by  iteratively solving $\mathbf{\mathcal{P}}^{'}_3(\mathbf{Q},\boldsymbol{\gamma,\theta})$.	
	\end{lemma}
	\begin{proof}
		The objective function of $\mathbf{\mathcal{P}}^{'}_3(\mathbf{Q},\boldsymbol{\gamma,\theta})$ is an  approximation of the objective function of  $\mathbf{\mathcal{P}}_3(\mathbf{Q},\boldsymbol{\gamma,\theta})$ in the neighborhood of $\mathbf{Q}$. To ensure the convergence and equivalence between the original optimized function $f(x)$ and its approximation function $u(\cdot,\cdot)$, four conditions need to be satisfied when applying SLM algorithm \cite{razaviyayn2013unified}:
		\begin{itemize}
			\item \quad $u(y,y)=f(y)\quad \forall y\in\mathcal{X}$,
			\item \quad $u(x,y)\ge f(x) \quad \forall x,y\in\mathcal{X}$,
			\item \quad $\left.u^{\prime}(x, y ; d)\right|_{x=y}=f^{\prime}(y ; d) \quad \forall d \text { with } y+d \in \mathcal{X}$ ,
			\item \quad $u(x,y)$ is continuous in $(x,y)$.
		\end{itemize}
	
		Notice that  $p (\mathbf{Q})$ is a convex function whose first order linear approximation is $p \left( {{\mathbf{Q}^\prime }} \right) + {\mathop{\rm tr}\nolimits} \left[ {\nabla p {{\left( {{\mathbf{Q}^\prime }} \right)}^T}\left( {\mathbf{Q} - {\mathbf{Q}^\prime }} \right)} \right]$ and $p (\mathbf{Q})$ is a convex function. We have
		\begin{equation}
			p (\mathbf{Q}) \ge p \left( {{\mathbf{Q}^\prime }} \right) + {\mathop{\rm tr}\nolimits} \left[ {\nabla p {{\left( {{\mathbf{Q}^\prime }} \right)}^T}\left( {\mathbf{Q} - {\mathbf{Q}^\prime }} \right)} \right].
		\end{equation}
		It is obvious that the linear approximation of $p \left(\mathbf{Q}\right)$ is globally lower than $p \left(\mathbf{Q}\right)$. Therefore, the approximation to \eqref{function9} is a lower bound concave approximation to \eqref{function8} satisfying the regulations on $u(\cdot,\cdot)$.
	\end{proof}
	Finally, $\mathbf{\mathcal{P}}^{'}_3(\mathbf{Q},\boldsymbol{\gamma,\theta})$ is a convex optimization function for a convex feasible set  $\left\{q_{k,j}\right\}$ under fixed value of $\boldsymbol{\gamma}$ and $\boldsymbol{\theta}$. 	Therefore, we then propose an iterative MUD workload distribution method. We initialize $\mathbf{Q}$ that meets $\mathbf{C}_1,\mathbf{C}_2,\mathbf{C}_4$. Compute  \eqref{function10} and solve $\mathbf{\mathcal{P}}_3(\mathbf{Q},\boldsymbol{\gamma,\theta})$. Then, we iteratively solve $\mathbf{\mathcal{P}}^{'}_3(\mathbf{Q},\boldsymbol{\gamma,\theta})$ until convergence. The above entire procedure repeats until convergence, which is summarized in  Algorithm 1.
	
	Each step of Algorithm \ref{algorithm1} involves solving convex optimization problems $\mathbf{\mathcal{P}}_3(\mathbf{Q},\boldsymbol{\gamma,\theta})$ and $\mathbf{\mathcal{P}}^{'}_3(\mathbf{Q},\boldsymbol{\gamma,\theta})$. Assume that each convex problem can be solved using primal-dual interior point algorithm thus the complexity is $\mathcal{O}\left(\sqrt{n} \log _2\left(\frac{1}{\varepsilon}\right)\right)$ with $n$ the number of inequality constraints. Thus the overall complexity of the proposed allocation method is $\mathcal{O}\left(\sqrt{K+J+KJ} \log _2\left(\frac{1}{\varepsilon_n}\right)\right)+\mathcal{O}\left( \log _2\left(\frac{1}{\varepsilon_m}\right)\right)$. The convergence of the proposed allocation method is guaranteed due to equivalent transformation and {\bf Lemma \ref{lemma2}}.

		\begin{table}[t]
		\centering
		\caption{Simulation settings.}
		\label{parameters}
		\renewcommand{\arraystretch}{1.1}
		\begin{tabular}{c c}
			\toprule
			\textbf{Definition} & \textbf{Value} \\ [0.5ex]
			\hline
			Frequency band          & 2 GHz   \\
			Frequency bandwdith        & 15 MHz  \\ 
			G/T & -33.6 dB/K  \\
			Eequivalent isotropically radiated power & $-7$ dBW \\
			Noise power density    & $-173$ dBm/Hz   \\ 
			Low earth orbit satellite altitute & 600 km \\
			\bottomrule
		\end{tabular}
	\end{table}
	
	\section{Numerical Results}
	In this section, numerical results are presented to evaluate the performance of the proposed  method.  We consider $K = 32$, $N = 12$, $J=8$ and the same spreading sequences $\mathbf{s}_k$ as \cite{yuan2020iterative}. In the following, we consider $q_l=\frac{q_s\times K}{J}$ for without loss of generality. The same channel model used in Section 6.6.2 of 3GPP Technical Report 38.811 \cite{3gpp.38.811} is adopted and other
	simulation parameters are summarized in Table \ref{parameters}.
	
	Figs. \ref{convergence1} and \ref{convergence2} give the convergence performance of the proposed method under different values of $\epsilon$ and $q_s$.  It can be observed that the proposed  method has good convergence and converges within $10$ external iterations.
	\begin{figure}[!t] 
	\centering
	\captionsetup[subfloat]{labelsep=none,format=plain,labelformat=empty,textfont=rm}
	\subfloat[\footnotesize{{\color{black}(a)  Overall sum-rate vs. iteration numbers with $q_s= 3$}}]{
		\begin{minipage}{0.24\textwidth}
			\label{convergence1}
			\centering 
			\includegraphics[width=1.0\linewidth]{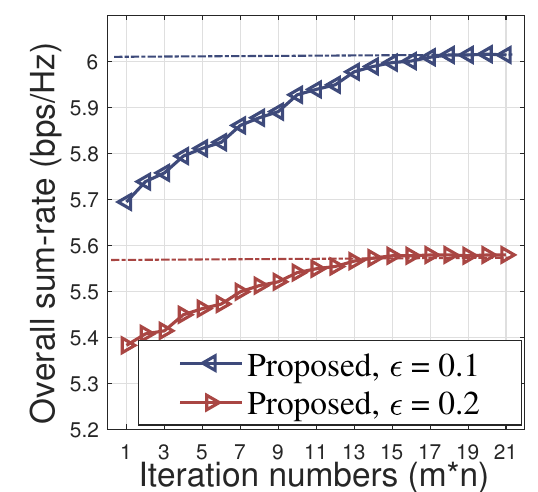}
		\end{minipage}
	}
	\subfloat[\footnotesize{{\color{black}(b)   Overall sum-rate vs. iteration numbers with $\epsilon= 0.2$ }}]{
		\begin{minipage}{0.24\textwidth}
			\label{convergence2}
			\centering 
			\includegraphics[width=1.0\linewidth]{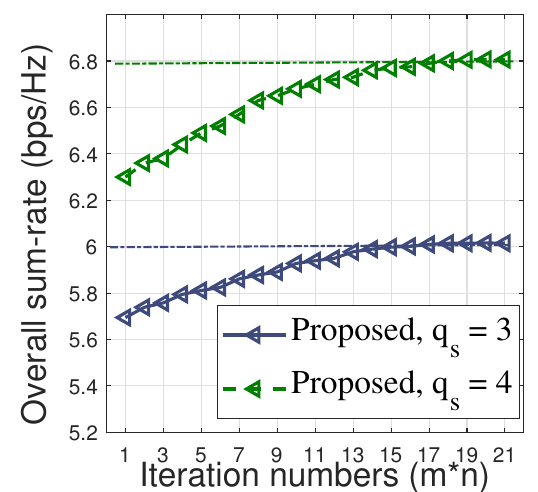}
		\end{minipage}
	}
	{\color{black}\caption{ Convergence performance: overall sum-rate of the proposed method vs. iteration numbers with varying (a)  $\epsilon$ (b)  $q_s$.}}
\end{figure}	
		\begin{figure}[!t] 
		\centering
		\captionsetup[subfloat]{labelsep=none,format=plain,labelformat=empty,textfont=rm}
		\subfloat[\footnotesize{{\color{black}(a)  Overall sum-rate vs. $\epsilon$ with $q_s=3$}}]{
			\begin{minipage}{0.24\textwidth}
				\label{benchmark1}
				\centering 
				\includegraphics[width=1.0\linewidth]{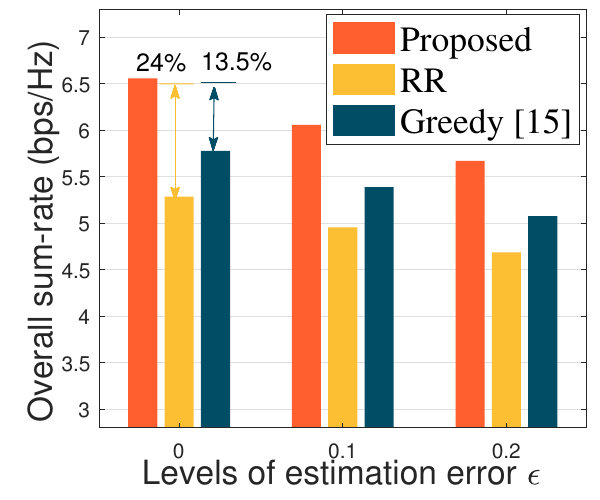}
			\end{minipage}
		}
		\subfloat[\footnotesize{{\color{black}(b)   Overall sum-rate vs. $q_s$ with $\epsilon$=0.2}}]{
			\begin{minipage}{0.24\textwidth}
				\label{benchmark2}
				\centering 
				\includegraphics[width=1.0\linewidth]{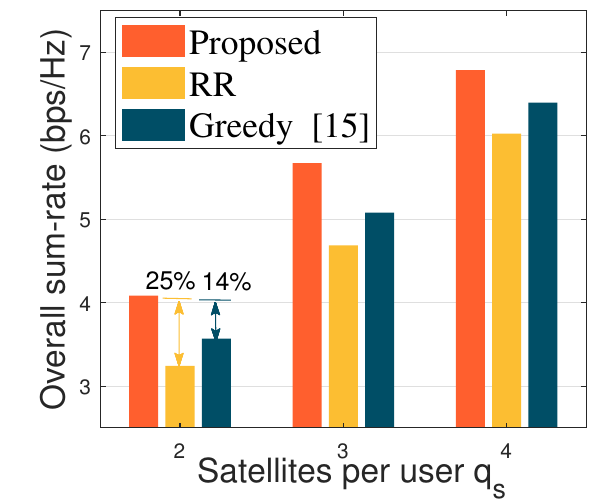}
			\end{minipage}
		}
		{\color{black}\caption{Comparisons on the sum-rate with varying:  (a) $\epsilon$, and (b) $q_s$.}}
	\end{figure}

	Fig. \ref{benchmark1} compares the performance of the proposed method with some benchmarks. Greedy method is performed by allowing each user to choose those satellites with the best channel qualities \cite{6519408}. Round robin (RR) method means the task allocation among $J$ satellites are decided in a round robin manner randomly to meet $\mathbf{C}_1,\mathbf{C}_2,\mathbf{C}_3$. Compared to greedy method and RR method, our approach provides another $24\%$ and $13.5\%$ gain when the interference from non-target users is canceled perfectly. Also, it can be observed that the overall sum-rate decreases with the increase of estimation errors of non-target users $\epsilon$ for each satellite.	
	\begin{figure}[!t] 
		\centering
		\captionsetup[subfloat]{labelsep=none,format=plain,labelformat=empty,textfont=rm}
		\subfloat[\footnotesize{{\color{black}(a)  Overall sum-rate vs. $J$ with $\epsilon= 0.4$}}]{
			\begin{minipage}{0.24\textwidth}
				\label{availableset1}
				\centering 
				\includegraphics[width=1.0\linewidth]{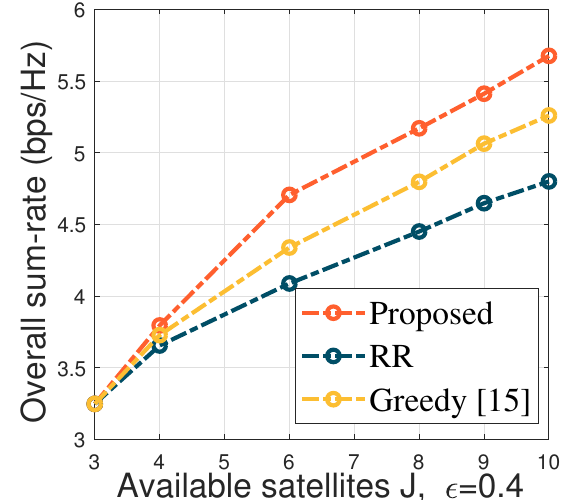}
			\end{minipage}
		}
		\subfloat[\footnotesize{{\color{black}(b)   Overall sum-rate vs. $J$ with $\epsilon= 0.6$ }}]{
			\begin{minipage}{0.24\textwidth}
				\label{availableset2}
				\centering 
				\includegraphics[width=1.0\linewidth]{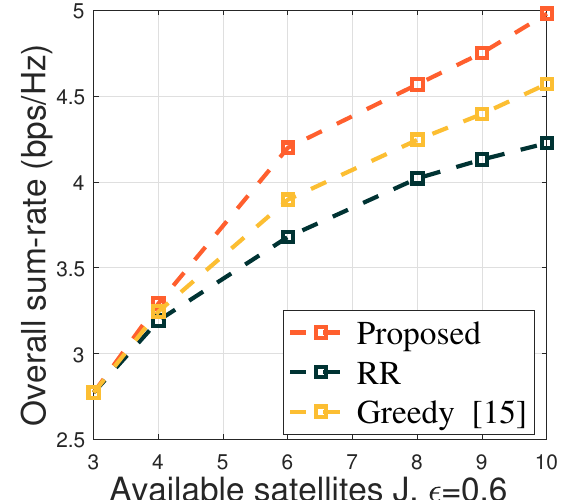}
			\end{minipage}
		}
		{\color{black}\caption{ Comparisons on the sum-rate with varying  $J$ under (a) $\epsilon=0.4$, and (b) $\epsilon=0.6$.}}
	\end{figure}	
	
	Fig. \ref{benchmark2} compares the performance of the proposed  method with the aforementioned methods with different values of $q_s$. It can be seen that all methods see increases of overall sum-rate as the increase of $q_s$. If more satellites are simultaneously serving a single user, it would result in higher system sum-rate. However, $q_l$ would also increase, resulting in MUD complexity for a single satellite.
	
	Figs. \ref{availableset1} and \ref{availableset2} compare the performance of the proposed method with the aforementioned methods with varying $J$ under different $\epsilon$. As  $J$ increases, the system has more choices and tends to choose those satellites with better channel conditions thus introducing better performance.
	
	Fig. \ref{MUD} compares both the processing load and the sum-rate gap of the proposed method and the centralized processing upper bound with varying $q_s$. The MUD processing load is approximated by the number of traversals calculated as stated in {\bf Proposition \ref{pro}}.
	As $q_s$ increases,  the pace of MUD processing load growth surpasses the rate of sum-rate enhancement significantly. Specially, around $75\%$ computational complexity can be reduced with less than  $10\%$ sum-rate loss performing under our proposed allocation method. 
		\begin{figure}
			\centering
			\includegraphics[width=1.0\linewidth]{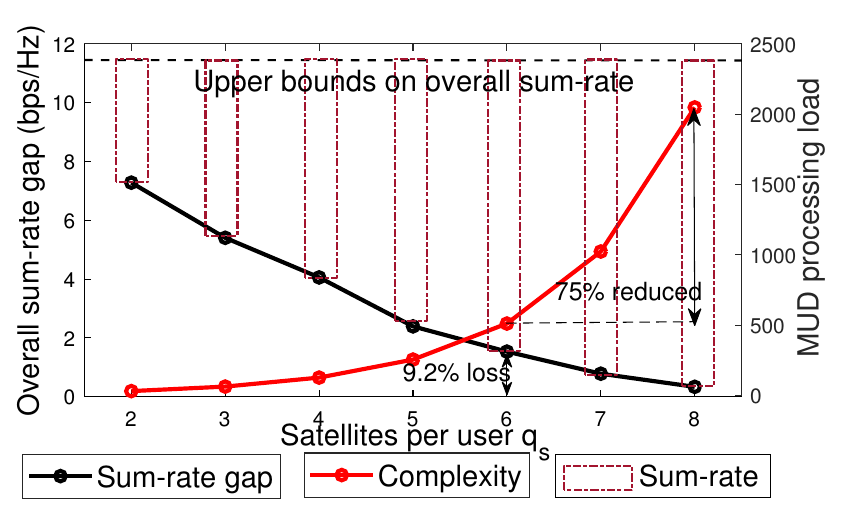}
			\caption{Comparisons on the sum-rate gap and the MUD processing load with varying $q_s$.}
			\label{MUD}
		\end{figure}
	\section{Conclusion}
	This paper studied the processing load allocation  for on-board MUD in payload-constrained satellite networks. We first formulated the load allocation problem as a system sum-rate maximization problem by optimizing the satellite-user matchings. An iterative algorithm using quadratic transformation and SLM algorithm was proposed to solve the non-trivial problem. 
	Numerical results showed remarkably complexity reduction compared with centralized processing as well as significant sum-rate gain compared with other allocation  methods.

	\bibliography{IEEEabrv,myref}
	\bibliographystyle{IEEEtran}

\end{document}